\documentclass[%
reprint,
aip,
jcp
]{revtex4-2}

\usepackage{graphicx}% Include figure files
\usepackage{dcolumn}% Align table columns on decimal point
\usepackage{bm}% bold math
\usepackage{graphicx}
\usepackage{tikz}
\usepackage{subfigure}
\usepackage{amsfonts,amsmath,amssymb,amsthm}

\newtheorem{theorem}{Theorem}
\newtheorem{lemma}{Lemma}

\begin{document}

%\preprint{APS/123-QED}

\title{
Representational power of 
selected neural network quantum states
in second quantization
}
\author{Zhendong Li}\email{zhendongli@bnu.edu.cn}
\address{Key Laboratory of Theoretical and Computational Photochemistry, Ministry of Education, College of Chemistry, Beijing Normal University, Beijing 100875, China}
\author{Tong Zhao}\email{zhaotong@ict.ac.cn}
\affiliation{State Key Lab of Processors, Institute of Computing Technology, Chinese Academy of Sciences, Beijing 100095, China}
\author{Bohan Zhang}
\address{Key Laboratory of Theoretical and Computational Photochemistry, Ministry of Education, College of Chemistry, Beijing Normal University, Beijing 100875, China}

\begin{abstract}
Neural network quantum states emerge as a promising tool for solving
quantum many-body problems. However, its successes and limitations
are still not well-understood in particular for Fermions with complex
sign structures. Based on our recent work [J. Chem. Theory Comput. 21, 10252-10262 (2025)], we generalizes the 
restricted Boltzmann machine to a more general class of states
for Fermions, formed by product of `neurons' and hence will be referred
to as neuron product states (NPS). NPS builds correlation in a very different way,
compared with the closely related correlator product states (CPS)
[H. J. Changlani, et al. Phys. Rev. B, 80, 245116 (2009)],
which use full-rank local correlators. In constrast, each correlator in NPS contains long-range correlations across all the sites, with its representational power constrained by the simple function form. 
We prove that products of such simple nonlocal correlators can approximate any wavefunction arbitrarily well under certain mild conditions on the form of activation functions. In addition, we also provide elementary proofs for the universal approximation capabilities of feedforward neural network (FNN) and neural network backflow (NNBF) in second quantization. Together, these results provide a deeper insight into the neural network representation of many-body wavefunctions in second quantization.
\end{abstract}

\maketitle

\section{Introduction}
Accurate and efficient simulation of quantum many-body problems on classical computers has been a long-standing challenge for computational physics and chemistry due to the exponential growth
of the size of the Hilbert space as system size increases.
During the past decades, a plethora of methods have been developed with their own advantages and disadvantages.
From the early days of quantum chemistry, the configuration interaction (CI) method\cite{shavitt1998history} is the most conceptually simple approach to treating correlated
electrons. Later, the hierarchy of coupled cluster (CC) methods\cite{shavitt2009many} 
becomes more dominant due to their size extensivity. 
These two methods are ''universal'', in the sense that any wavefunction can be approximated by increasing the
excitation rank. 
%In practice, they are usually truncated to low excitation levels in order to be tractable. 
Tensor network states\cite{orus2014practical} (TNS), which include matrix product states (MPS) as a representative, are better choices for strongly correlated systems\cite{chan2011density}. 
They are also universal, as long as the bond dimension
can be arbitrarily large. 
The universal approximation capability of a wavefunction ansatz is important, because it provides a solid theoretical guarantee
for its ultimate accuracy.

Recently, neural networks (NN) become an emerging technique for simulation of quantum many-body problems. Carleo and Troyer employed restricted Boltzmann machine\cite{carleo2017solving} (RBM),
a generative model that can learn a probability distribution over a set of inputs\cite{le2008representational,montufar2011refinements},
as a variational ansatz for interacting spin problems on lattices, and achieved good accuracy compared with the state-of-the-art TNS. Unlike
TNS, which encodes area-law entanglement efficiently,
RBM can describe volume-law states\cite{deng2017quantum}.
Later, other machine learning architectures have also been exploited for spin systems, including feedforward NN (FNN)\cite{cai2018approximating,choo2018symmetries},
convolutional NN (CNN)\cite{liang2018solving,choo2019two},
recurrent NN (RNN)\cite{hibat2020recurrent,hibat2022supplementing},
autoregressive NN\cite{sharir2020deep,barrett2022autoregressive}, 
neural network backflow (NNBF)\cite{luo2019backflow, liu2024unifying, Liu2024, Liu2025}, etc.
The success and practical limitations of these neural network
quantum states (NQS) are not fully understood yet, and are being actively studied\cite{westerhout2020generalization,bukov2021learning,park2022expressive,viteritti2022accuracy,sharir2022neural}.
Compared with spin systems, the application of neural networks  
for Fermions remains largely unexplored until very recently\cite{nomura2017restricted,yang2020artificial,
choo2020fermionic,yoshioka2021solving,barrett2022autoregressive,robledo2022fermionic,wu2023real}. In this work, we focus on Fermions
on lattices or electrons in molecules described
within the second quantization framework, 
and we refer the readers to Ref. \cite{hermann2023ab} for applications of NN in the first quantization.

\begin{figure*}[t]
    \centering
    \includegraphics[width=1.0\textwidth]{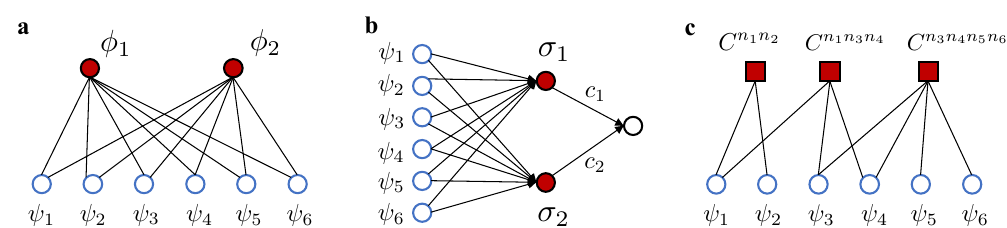}
    \caption{Illustration of different neural-network wavefunction ans\"atze for $K=6$: 
    (a) restricted Boltzmann machine (RBM) and neuron product states (NPS); 
    (b) feedforward neural networks (FNN);
    (c) correlator product states\cite{changlani2009approximating} (CPS). Each red circle in (a) and (b) represents a hidden neuron.}
    \label{fig:NPS}
\end{figure*}

Inspired by RBM, we recently introduce a more general class of states for fermions composed by product of 'neurons', which will be referred to as neuron product states\cite{wu2025hybrid} (NPS). 
NPS take the following general form
\begin{eqnarray}
\Psi_{\mathrm{NPS}}(n_1,\cdots,n_K) = \prod_{\alpha=1}^{N_h} \phi\left(b_\alpha + \sum_{k=1}^K W_{\alpha k} n_k\right),\label{eq:nps}
\end{eqnarray}
where $\Psi_{\mathrm{NPS}}(n_1,\cdots,n_K)$ is the wavefunction
in the occupation number representation, $n_k\in\{0,1\}$,
$K$ is the number of spin-orbitals, $\phi(x)$ is an activation function, $b_\alpha$ and $W_{\alpha k}$
are real parameters. Each factor $\phi$ in Eq. \eqref{eq:nps} is referred to as a neuron,
and $N_h$ is the number of neurons. 
We will denote
$\Psi_{\mathrm{NPS}}(n_1,\cdots,n_K)$ as $\Psi_{\mathrm{NPS}}(\vec{n})$
and omit the lower/upper limit of summation/product for brevity.
This form generalizes
the standard RBM for discrete probability, 
which after tracing out the hidden units can be written as
\begin{eqnarray}
P_{\mathrm{RBM}}(\vec{n})=e^{\sum_{k} a_k n_k}\prod_{\alpha}(1 + e^{b_\alpha + \sum_{k} W_{\alpha k}n_k}).\label{eq:rbm}
\end{eqnarray}
Note that each factor in RBM is positive. Carleo and Troyer\cite{carleo2017solving} use complex parameters
in RBM for wavefunctions that can take negative values.
A general $\phi(x)$ is used in NPS, which can take both positive and negative values. In Ref. \cite{wu2025hybrid}, we use
$\phi(x)=\cos(x)$ and \eqref{eq:nps} is used as a correlator
multiplied on another wavefunction ansatz. 
NPS can be pictorially represented in Fig. \eqref{fig:NPS}a.
In this work, we focus on NPS itself and ask a fundamental question that whether it can be a universal approximator.

Formally, we can also rewrite Eq. \eqref{eq:nps} as
\begin{eqnarray}
\Psi_{\mathrm{NPS}}(\vec{n}) = \exp\left[{\sum_\alpha \kappa\left(b_\alpha + \sum_{k}W_{\alpha k} n_k\right)}\right].\label{eq:nps2}
\end{eqnarray}
with $\kappa(x) = \ln\phi(x)$. The sum $\sum_a \kappa(b_\alpha + \sum_{k}W_{\alpha k} n_k)$ displays an apparent similarity with FNN.
Specifically, the FNN with one hidden layer, see Fig. \ref{fig:NPS}b,
can be written in a similar form as
\begin{eqnarray}
f_{\mathrm{FNN}}(\vec{x}) = \sum_{\alpha} c_\alpha \sigma(b_\alpha + \sum_{k} W_{\alpha k} x_k),\label{eq:FNN}
\end{eqnarray}
where $\sigma$ is an activation function. However, an important
difference is that there is no linear combination coefficients $c_\alpha$ in Eq. \eqref{eq:nps2}. Thus, the classical universal approximation theorem (UAT)\cite{funahashi1989approximate,cybenko1989approximation,hornik1989multilayer,hornik1991approximation,leshno1993multilayer} for multilayer FNN cannot be applied. Similarly, we find that the proof RBM as 
universal approximators of discrete distributions\cite{le2008representational,montufar2011refinements,huang2021neural} cannot be applied to NPS \eqref{eq:nps}, as these proofs\cite{le2008representational,huang2021neural} hinge on 
the form of factors
$1+e^x$ in Eq. \eqref{eq:rbm}, such that one can modify 
a particular amplitude of wavefunction by choosing appropriate $b_{\alpha}$ and $W_{\alpha k}$.

In this work, we focus on activation functions that are \emph{analytic} in certain domain, which means that they can be represented by Taylor series locally, i.e., $\phi(x)=\sum_{n=0}^{\infty}\frac{\phi^{(n)}(x_0)}{n!}(x-x_0)^n$.
Typical examples include elementary functions (polynomials,
exponential, trigonometric, etc.) and piecewise
defined functions. As in neural networks,
they are the most commonly used forms of activation functions\cite{ramachandran2017searching}, e.g.,
logistic sigmoid $1/(1+e^{-x})$, hyperbolic tangent $\tanh(x)$, and ReLU\cite{glorot2011deep} (rectified linear unit). 
We will show that the necessary and sufficient conditions for NPS \eqref{eq:nps} as a universal approximator
for quantum states is that (1) $\phi(x)$ can change signs
and (2) the logarithm of the activation function $\kappa(x)=\ln\phi(x)$ is not a polynomial with degree less than the number of orbitals $K$.
This result generalizes the universal property of RBM
to a more general class of variational ansatz.

Before presenting the proof, we discuss the connections of NPS with other
variational ans\"atze. A closely related class of wavefunctions is the
correlator product states (CPS)\cite{changlani2009approximating,clark2018unifying}
(including Jastrow factor and entangled plaquette states\cite{mezzacapo2009ground,glasser2018neural} as
special cases), which express the wavefunction as a product of correlators.
The tensors $C^{n_1n_2}$ and $C^{n_1n_3n_4}$ in Fig. \ref{fig:NPS}c
are examples of local two-site and nonlocal three-site correlators, respectively.
A typical two-site CPS reads
\begin{eqnarray}
\Psi_{\mathrm{CPS}}(\vec{n})=\prod_{i<j}C^{n_in_j}.\label{eq:CPS}
\end{eqnarray}
In the limit of $K$-site correlators, CPS becomes exact.
We can compare NPS with CPS, which builds the correlation in a drastically different way. Each term in the product \eqref{eq:nps} contains long-range correlations across all the sites, with the representational power
constrained by the simple function form. Thus, a 
number of products are needed to accurately describe the
wavefunction. Another interesting connection is that,
by restricting the sum in $\sum_k W_{\alpha k}n_k$ in Eq. \eqref{eq:nps}
to certain sites, a UAT for NPS will imply that every correlator can be approximated by products of $\phi$. For instance, the four-site correlator $C^{n_3n_4n_5n_6}$ in Fig. \ref{fig:NPS}c can be approximated arbitrarily well by products of functions of form $\phi(b_\alpha + \sum_{k\in \{3,4,5,6\}} W_{\alpha k}n_k)$. This connection will be clearer in the proof of UAT for NPS.

The remaining part of the paper is organized as follows.
In Secs. \ref{sec:FNN} and \ref{sec:NNB}, 
we provide elementary proofs for UAT of FNN and NNBF in second quantization, respectively, which introduce
basic notations and important techniques 
for proving UAT for NPS in Sec. \ref{sec:NPS} for activation functions
under simple conditions. The proof of UAT
for NPS for general activation functions is more
technically involved and hence is presented 
in Appendix.
Conclusion and outlook are given in the last section.

\section{Feedforward neural network (FNN) states}\label{sec:FNN}
Applying FNN in Eq. \eqref{eq:FNN} for many-body wavefunction
in second quantization leads to the FNN ansatz\cite{cai2018approximating,choo2018symmetries}
\begin{eqnarray}
\Psi_{\mathrm{FNN}}(\vec{n}) = 
\sum_{\alpha} c_\alpha \sigma(b_\alpha+\vec{w}_\alpha^T \vec{n}),\label{eq:FNN1}
\end{eqnarray}
or more compactly
\begin{eqnarray}
\Psi_{\mathrm{FNN}}(\vec{n}) = 
\vec{c}^T \sigma(\vec{b} + W \vec{n}),\label{eq:FNN2}
\end{eqnarray}
where $\vec{c}\in\mathbb{R}^{N_h}$, $W\in\mathbb{R}^{N_h\times K}$,
$\vec{b}\in\mathbb{R}^{N_h}$ with $N_h$ being the number of hidden neurons.
The $2^{K}$ occupation number vectors $\vec{n}$ form a Boolean hypercube $\mathcal{B}^K=\{0,1\}^K$. We denote the dimension of the Fock space by $D_K=2^K$. We assume that the activation function $\sigma$
satisfies
\begin{eqnarray}
    \lim_{x\rightarrow -\infty} \sigma(x)=0,\quad
    \lim_{x\rightarrow +\infty} \sigma(x)=1.
\end{eqnarray}
A typical example is the sigmoid function $\sigma(x)=\frac{1}{1+e^{-x}}$.

\begin{theorem}\label{thm:FMM}
Given sufficiently large $N_h$, the FNN ansatz \eqref{eq:FNN1} is universal for representing wavefunction
in second quantization, denoted as a vector by $\vec{\Psi} \in \mathbb{R}^{D_K}$ in the occupation number representation, in the sense that for any wavefunction $\Psi:\mathcal{B}^K \rightarrow \mathbb{R}$
and any $\epsilon>0$, there exists a FNN such that
the network output $\Psi_{\mathrm{FNN}}(\vec{n})$ satisfies
$|\Psi_{\mathrm{FNN}}(\vec{n})-\Psi(\vec{n})|<\epsilon$ for 
every $\vec{n} \in \mathcal{B}^K$.
\end{theorem}

In fact, this follows directly from the UAT for functions with continuous variables, by embedding the Boolean hypercube $\mathcal{B}^K$ into $\mathbb{R}^K$. But we give a very elementary proof of it for discrete variables using the following lemma.

\begin{lemma}\label{lemma:ziz}
Let $\vec{z}=2\vec{n}-\vec{e} \in \{-1,1\}^K$ with $\vec{e}=(1,\cdots,1)^T$, then 
for $\vec{z}\in\{\vec{z}_j\}_{j=0}^{D_K-1}$, 
$\vec{z}_i\cdot\vec{z} \in \{K,K-2,\cdots,-(K-2),-K\}$.
The maximal value $K$ is reached for $\vec{z}=\vec{z}_i$,
while the minimal value $-K$ is reached for $\vec{z}=-\vec{z}_i$. 
\end{lemma}

Without the loss of generality, the $D_K$ vectors $\vec{n}_j$ and $\vec{z}_j$ are assumed to be labeled lexicographically. This lemma can be verified by simple calculations. Now we prove Theorem \ref{thm:FMM}.

\begin{proof}
Introduce a matrix $F\in\mathbb{R}^{N_h\times D_K}$ with elements
\begin{eqnarray}
    F_{\alpha j} = \sigma(b_\alpha  + \vec{w}^T_\alpha  \vec{n}_j).
\end{eqnarray}
We want to show that with $N_h=2^K$, 
$F$ can be made arbitrarily close to an identity matrix
by choosing $(\vec{b},W)$ appropriately. Then, $\vec{c}=\vec{\Psi}$ showing that the weights of FNN store the
wavefunction in such case.

For each $i$, we have $\vec{z}_i^T(\vec{z}-\vec{z}_i) \in \{0,-2,\cdots,-2K\}$ using Lemma \ref{lemma:ziz}
and $\vec{z}_i^T\vec{z}_i=K$,
such that $\vec{z}_i^T(\vec{z}-\vec{z}_i)+1 
= 2\vec{z}_i^T \vec{n} - \vec{z}_i^T\vec{e} - K + 1
\in \{1,-1,\cdots,-2K+1\}$. 
By choosing $\vec{w}_i=2\theta\vec{z}_i$,
$b_i = \theta(- \vec{z}_i^T \vec{e}-K+1)$, and let 
$\theta\rightarrow +\infty$, we obtain
\begin{eqnarray}
\vec{w}_i^T\vec{n}_j+b_i \rightarrow +\infty,\quad \sigma(\vec{w}_i^T\vec{n}_j+b_i) \rightarrow 1,\quad j=i,\\
\vec{w}_i^T\vec{n}_j+b_i \rightarrow -\infty,\quad \sigma(\vec{w}_i^T\vec{n}_j+b_i) \rightarrow 0,\quad j\ne i.
\end{eqnarray}
Then, $F$ can be made arbitrarily close to an identity matrix; hence the FNN ansatz is universal up to arbitrary precision.
\end{proof}

\section{Neural network backflow (NNBF) states}\label{sec:NNB}
The NNBF ansatz\cite{luo2019backflow, liu2024unifying, Liu2024, Liu2025}
in second quantization is defined by
\begin{eqnarray}
\Psi_{\mathrm{NNBF}}(\vec{n}) = \det[\phi_{p_k m}(\vec{n})],\label{eq:NNBF}    
\end{eqnarray}
where $m\in\{1,\cdots,N\}$ with $N$ being
the number of electrons, $p_k$ ($k\in \{1,\cdots,N\}$) 
represent the indices of occupied orbitals in $\vec{n}$, and
$\phi_{pm}(\vec{n})$ can be viewed as a set of
configuration-dependent orbitals generated by a FNN via
\begin{eqnarray}
\phi_{pm}(\vec{n}) = 
\vec{c}_{pm}^T \sigma(\vec{b} + W\vec{n}),\label{eq:NNBForb}
\end{eqnarray}
where $\vec{c}_{pm} \in \mathbb{R}^{N_h}$ and $p\in\{1,\cdots,K\}$.

\begin{theorem}
The NNBF ansatz \eqref{eq:NNBF} is universal for sufficiently large $N_h$.
\end{theorem}

\begin{proof}
By the proof in the previous theorem, one can choose $N_h=D_K$ and
$(b_\alpha,\vec{w}_\alpha)$ such that for given $\vec{n}_i$, 
only one term in the summation \eqref{eq:NNBForb} 
contributes 
\begin{eqnarray}
\phi_{pm}(\vec{n}_i) = c_{pm,i},
\end{eqnarray}
then the wavefunction amplitude is simply
\begin{eqnarray}
\Psi_{\mathrm{NNBF}}(\vec{n}_i)=
\det\left[
\begin{array}{cccc}
c_{p_1 1,i} & c_{p_1 2,i} & \cdots & c_{p_1 N,i} \\
c_{p_2 1,i} & c_{p_2 2,i} & \cdots & c_{p_2 N,i} \\
\vdots & \vdots & \ddots & \vdots \\
c_{p_N 1,i} & c_{p_N 2,i} & \cdots & c_{p_N N,i} 
\end{array}
\right].
\end{eqnarray}
By choosing $c_{p_1 1,i}=\Psi(n_i)$, $c_{p_m m,i}=1$ for $m=2,\cdots,N$,
and $c_{p_k m,i}=0$ for other entries, it is seen that the NNBF ansatz
is universal.
\end{proof}

Remarks: For a single (Hartree-Fock) determinant ansatz, 
$c_{p_k m,i}$ is independent of $i$.

\section{Neural network product states (NPS)}\label{sec:NPS}
We assume that the target wavefunction 
$\Psi(\vec{n})$ is real, and 
all the parameters $(W,\vec{b})$ are also real in
NPS defined in Eq. \eqref{eq:nps}.
Besides, we assume that activation function $\sigma$ is in $(-1,1)$ and satisfies
\begin{eqnarray}
    \lim_{x\rightarrow +\infty} \sigma(x)=1.\label{eq:nps-condition}
\end{eqnarray}
A typical example is $f(x)=\tanh(x)$. Under such condition,
we can give an elementary proof that

\begin{theorem}\label{thm:nnp}
The NPS ansatz with activation function $\sigma\in(-1,1)$ and
satisfying Eq. \eqref{eq:nps-condition}
is universal for sufficiently large $N_h$.
\end{theorem}

In fact, NPS with more general activation function can also be proved universal.

\begin{theorem}\label{thm:NPSgeneral}
The NPS ansatz is universal if and only if
\begin{enumerate}
\item $\phi(x)$ can produce both positive and negative values (i.e., $\exists x_1, x_2$ such that $\phi(x_1) > 0$ and $\phi(x_2) < 0$),
\item $\ln\phi(x)$ is not a polynomial of degree less than $K$.
\end{enumerate}
\end{theorem}

However, its proof is more technically involved, and hence we present
it in Appendix. The proof of Theorem \ref{thm:nnp} is simpler based
on the following lemma.

\begin{lemma}[Gordan's lemma]
Let $A\in \mathbb{R}^{m\times n}$ be a matrix. Then exactly one of the following
statement is true:
\begin{enumerate}
\item There exists $\vec{x}\in\mathbb{R}^n$ such that $A\vec{x}>0$ (componentwise).
\item There exists $\vec{y}\in\mathbb{R}^m$, $\vec{y}\ne \vec{0}$, with $y_k\ge 0$ (componentwise)
such that $A^T\vec{y}=0$.
\end{enumerate}
\end{lemma}
This is a fundamental result in linear algebra and convex analysis\cite{perng2015note}. Using it, we first prove the following
useful result.

\begin{lemma}\label{thm:vset}
Let $V$ denotes the set of vectors 
\begin{eqnarray}
V = \{\vec{n}_j - \vec{n}_i \,:\,j\ne i\},\label{eq:Vset}
\end{eqnarray}
then there exists a vector $\vec{u}\in\mathbb{R}^K$ such that $\vec{u}^T\vec{v}>0,\,\forall \vec{v}\in V$.
\end{lemma}

\begin{proof}
If no such $\vec{u}$ exists, then by Gordan's lemma, case 2 holds, that is,
there exist nonnegative coefficients $\lambda_j \ge 0$ ($j\ne i$, not all zero) such that
\begin{eqnarray}
\sum_{j\ne i} \lambda_j(\vec{n}_j - \vec{n}_i)=\vec{0}. \label{eq:lambdaj1}
\end{eqnarray}
We now show that this will lead to contradictions.

Let $S=\sum_{j\ne i}\lambda_j$. Since not all $\lambda_j$ are zero, we have $S>0$.
Rearranging Eq. \eqref{eq:lambdaj1}, we obtain
\begin{eqnarray}
\sum_{j\ne i}\lambda_j \vec{n}_j = S\vec{n}_i.    
\end{eqnarray}
Consider the $p$-th component of this equation, if $(\vec{n}_i)_p=1$, then $\sum_{j\ne i}\lambda_j (\vec{n}_j)_p = S$. Since
$\sum_{j\ne i}\lambda_j (\vec{n}_j)_p \le 
\sum_{j\ne i}\lambda_j = S$,  we have $(\vec{n}_j)_p=1$ for all $j$ with $\lambda_j>0$.
Likewise, if $(\vec{n}_i)_p=0$, then $\sum_{j\ne i}\lambda_j (\vec{n}_j)_p = 0$,
suggesting that for all $j$ with $\lambda_j>0$, $(\vec{n}_j)_p=0$.
In sum, for every $j$ with $\lambda_j>0$, Eq. \eqref{eq:lambdaj1} implies that $\vec{n}_j=\vec{n}_i$, which contradicts the fact that $j\ne i$.
Therefore, our assumption that no such $\vec{u}$ exists must be false.
\end{proof}

Remarks: The set $V$ is a finite set of nonzero vectors, and it does not contain any pair of opposite vectors. Because if $\vec{v}$ and $-\vec{v}$ were
both in $V$, then there exists $j$ and $k$, such that $\vec{n}_j-\vec{n}_i
=-\vec{n}_k + \vec{n}_i$, which implies that $\vec{n}_j+\vec{n}_k = 2\vec{n}_i$.
This is impossible for binary vectors $\vec{n}_j\ne \vec{n}_i$ and $\vec{n}_k\ne \vec{n}_i$. Geometrically, since $V$ is finite and contains no opposite vectors, the convex cone generated by $V$ does not contain the origin. Therefore, there exists a hyperplane through the origin that separates the origin from $V$, meaning all vectors in $V$ are on one side of the hyperplane. The normal vector to this hyperplane (pointing towards $V$) can be taken as $\vec{u}$, so that $\vec{u}^T\vec{v}>0$ for all $\vec{v}\in V$. 

Based on Lemma \ref{thm:vset}, we can prove Theorem \ref{thm:nnp}.
\begin{proof}
Let $N_h=D_K$. The key idea is to construct $N_h$ functions
$\{g_i\}$, where $g_i(\vec{n})=\phi(\vec{w}_i^T\vec{n}+b_i)$
with appropriately chosen parameters,
such that for each $i$: (1) $g_i(\vec{n}_i)\rightarrow\Psi(\vec{n}_i)$
and (2) $g_i(\vec{n}_j)\rightarrow 1$ for $j\ne i$.

The first part of the above statement is simple to prove. Since $\Psi(\vec{n}_i) \in [-1,1]$ and $\phi$ is continuous with range $(-1,1)$, there exists $x_0\in\mathbb{R}$ such that $\phi(x_0)$ is arbitrarily close to $\Psi(\vec{n}_i)$.

For the second part, we consider the set of vectors \eqref{eq:Vset}.
By Lemma \eqref{thm:vset}, there exists a vector $\vec{u}_i\in \mathbb{R}^K$ such that
$\vec{u}^T_i\vec{v}>0,\,\forall \vec{v}\in V$. Then, define 
\begin{eqnarray}
    \vec{w}_i = \theta \vec{u}_i,\quad b_i = x_0 - \vec{w}_i^T \vec{n}_i,
\end{eqnarray}
where $\theta>0$ is a scaling factor to be chosen large. We have
$\vec{w}_i^T \vec{n}_i+b_i = x_0$ and for $j\ne i$,
\begin{eqnarray}
    \vec{w}_i^T \vec{n}_j+b_i = x_0 + \theta \vec{u}^T_i(\vec{n}_j - \vec{n}_i).
\end{eqnarray}
Since $\vec{u}^T_i(\vec{n}_j - \vec{n}_i)>0$, by taking $\theta$ sufficiently large,
we can make $\vec{w}_i^T \vec{n}_j+b_i$ large enough, such that
$\phi(\vec{w}_i^T \vec{n}_j+b_i)$ is arbitrarily close to 1.

Then, by multiplying $N_h$ such functions $g_i(\vec{n})$ in Eq. \eqref{eq:nps}, the NPS ansatz can approximate any wavefunction arbitrarily well.
\end{proof}

The above proof relies on the condition \eqref{eq:nps-condition},
and hence does not generalize to the activation function $f(x)=\cos(x)$ used in Ref. \cite{wu2025hybrid}. In Appendix, we give a more general proof for Theorem \ref{thm:NPSgeneral}, which also reveals a deeper
connection with CPS.

\section{Conclusion}
In this work, we generalize RBM to a more broad class of 
states and prove the universal approximation capability 
of NPS given suitable active function. 
This lays the foundation for future exploration of NPS for strongly correlated fermions.
% Since RBM can represent volume law states\cite{deng2017quantum},
% while MPS can only describe area law state for comparison, 
% NPS can be useful for solving strongly correlated systems.
While in our previous work\cite{wu2025hybrid} we used such functions
with $\phi(x)=\cos(x)$ as correlators to enhance the expressivity of other variational ans\"atze, the present work demonstrate that the NPS itself is also a valid variational ansatz, which can be optimized using variational Monte Carlo\cite{mcmillan1965ground,gubernatis2016quantum,becca2017quantum}.
We expect the choice of activation function $\phi$ will be important
for NPS in practice, and the combination with composition of functions in
deep learning\cite{lecun2015deep} for more efficient representation 
is also promising.
These open questions will be explored in future,
which can extend RBM to more challenging systems.

\section*{Acknowledgment}
This work was supported by the Quantum Science and Technology-National Science and Technology Major Project(2023ZD0300200) and the Fundamental Research Funds for the Central Universities.

\section*{Appendix: Proof of Theorem \ref{thm:NPSgeneral} for NPS with general activation functions}\label{sec:NPSgeneral}
\subsection{Representation of $\Psi$ by multilinear polynomials}
We introduce a special representation of the many-body wavefunction by interpreting $\Psi(\vec{n})$ as a pseudo-Boolean function\cite{crama2011boolean}, which defines a mapping from the Boolean hypercube $\mathcal{B}^K$ to a field $F$ ($\mathbb{R}$ or $\mathbb{C}$). For convenience, we will work with the variable $z_k=-2n_k+1\in\{1,-1\}$. Using the inversion $n_k=\frac{1-z_k}{2}$, the same wavefunction but expressed in $\{z_k\}$ can be found as
\begin{eqnarray}
\Phi(z_1,\cdots,z_K)=
\Psi(\frac{1-z_1}{2},\cdots,\frac{1-z_k}{2}).\label{eq:z2n}
\end{eqnarray}
It can be represented uniquely as
a multilinear polynomial\cite{crama2011boolean}
\begin{eqnarray}
\Phi(z_1,\cdots,z_K) = \sum_{x_1\cdots x_K} \hat{\Phi}^{x_1\cdots x_K}z_1^{x_1}\cdots z_K^{x_K},\label{eq:multilinear}
\end{eqnarray}
where $x_k\in\{0,1\}$, $\hat{\Phi}^{x_1\cdots x_K}$ is a tensor with $2^K$ elements, which can
be called as the "Fourier coefficients" of the function $\Phi$.
Eq. \eqref{eq:multilinear} can be proved can realizing that
each factor 
$z_k^{x_k}=\left\{\begin{array}{cc}
   1,  &  z_k = 1 \\
  (-1)^{x_k}, &  z_k = -1
\end{array}\right.$ can be viewed as an element of
the Hadamard matrix 
$H_k=\frac{1}{\sqrt{2}}\left[\begin{array}{cc}
   1  & 1 \\
   1  & -1
\end{array}\right]$ (up to a constant factor). 
Then, we can rewrite 
Eq. \eqref{eq:multilinear} in a compact form
\begin{eqnarray}
\Phi = 2^{K/2} H^{\otimes K} \hat{\Phi},\label{eq:wht}
\end{eqnarray}
where both $\Phi$ and $\hat{\Phi}$ are viewed as vectors
and $H^{\otimes K}=H_1\otimes \cdots\otimes H_K$
is the Walsh-Hadamard transform. By the involutory property $H_k^2=I$, 
Eq. \eqref{eq:wht} can be inverted as
\begin{eqnarray}
\hat{\Phi} = 2^{-K/2} H^{\otimes K} \Phi,\label{eq:inversion}
\end{eqnarray}
which shows that $\hat{\Phi}$ is uniquely determined by $\Phi$.
Eq. \eqref{eq:inversion} can be written explicitly as
\begin{eqnarray}
\hat{\Phi}^{x_1\cdots x_K} = 2^{-K}\sum_{z_1\cdots z_K}z_1^{x_1}\cdots z_K^{x_K}
\Phi(z_1,\cdots,z_K).
% \label{eq:multilinear}
\end{eqnarray}
In the terminology of quantum computing\cite{nielsen2010}, $\hat{\Phi}$ is just the
wavefunction in the eigenbasis $\{|+\rangle,|-\rangle\}$ of
the Pauli $X$ operator for each qubit, while $\Phi$ is the
wavefunction in the computational basis $\{|0\rangle,|1\rangle\}$.

\subsection{Representation of functions $f(b+\vec{w}^T\vec{z})$}\label{sec:phi}
The above result shows that any discrete function: $\Psi:\{+1,-1\}^{n}\rightarrow F$
with the field $F=\mathbb{R}$ or $\mathbb{C}$ can be represented by a multilinear polynomial. Apply this result to the function of the form $f(b+\vec{w}^T\vec{z})$, where $f$
can be $\phi$ in Eq. \eqref{eq:nps} or $\kappa=\ln\phi$ in Eq. \eqref{eq:nps2}, leads to
\begin{eqnarray}
f(b+\vec{w}^T\vec{z}) = 2^{K/2} H^{\otimes K} \hat{f}.  
\end{eqnarray}
The inversion \eqref{eq:inversion} gives
\begin{eqnarray}
\hat{f}^{x_1 x_2\cdots x_K} = 2^{-K}
\sum_{z_1 z_2 \cdots z_K} 
z_1^{x_1} z_2^{x_2} \cdots z_K^{x_K}
f(b+\vec{w}^T \vec{z}).\label{eq:fkappa}
\end{eqnarray}
To better illustrate this formula, we give two concrete examples for
$K=1$ and $K=2$. For $K=1$,
\begin{eqnarray}
\hat{f}^0 &=& \frac{1}{2}[f(b+w_1)+f(b+w_1)], \\
\hat{f}^1 &=& \frac{1}{2}[f(b+w_1)-f(b-w_1)],
\end{eqnarray}
and for $K=2$,
\begin{widetext}
\begin{eqnarray}
\hat{f}^{00} &=& \frac{1}{4}[f(b+w_1+w_2)+f(b-w_1+w_2)+f(b+w_1-w_2)+f(b-w_1-w_2)], \\
\hat{f}^{10} &=& \frac{1}{4}[f(b+w_1+w_2)-f(b-w_1+w_2)+f(b+w_1-w_2)-f(b-w_1-w_2)], \\
\hat{f}^{01} &=& \frac{1}{4}[f(b+w_1+w_2)+f(b-w_1+w_2)-f(b+w_1-w_2)-f(b-w_1-w_2)], \\
\hat{f}^{11} &=& \frac{1}{4}[f(b+w_1+w_2)-f(b-w_1+w_2)-f(b+w_1-w_2)+f(b-w_1-w_2)].
\end{eqnarray}
\end{widetext}

In the later part, we will use the asymptotic behavior of $\hat{f}$ for small $\vec{\omega}$. To analyze this, we first use the Taylor expansion
\begin{eqnarray}
f(b+\vec{w}^T \vec{z}) = \sum_{m_k \ge 0}
\frac{1}{m_1!\cdots m_K!}f^{\sum_k m_k}(b) \prod_{k} \omega_k^{m_k} z_k^{m_k},\label{eq:taylor}
\end{eqnarray}
where $f^{\sum_k m_k}(b)$ represents the $\sum_k m_k$-th order derivative of $f(x)$
taken at $x=b$. It deserves to point that although Eq. \eqref{eq:taylor} looks quite similar
to Eq. \eqref{eq:multilinear}, they are different.
To be more precise, we can find $\hat{f}$ as
\begin{widetext}
\begin{eqnarray}
\hat{f}^{x_1 x_2\cdots x_K} &=& \frac{1}{2^K}
\sum_{z_1z_2\cdots z_K} 
z_1^{x_1}\cdots z_K^{x_K}
\sum_{m_k \ge 0}
\frac{1}{m_1!\cdots m_K!}f^{\sum_k m_k}(b) \prod_{k} \omega_k^{m_k} z_k^{m_k} \nonumber\\
&=&
\frac{1}{2^K}
\sum_{m_k \ge 0}
\frac{1}{m_1!\cdots m_K!}f^{\sum_k m_k}(b) \prod_{k} \omega_k^{m_k} 
\sum_{z_k} 
z_k^{x_k+m_k} \nonumber\\
&=&
\frac{1}{2^K}
\sum_{m_k \ge 0}
\frac{1}{m_1!\cdots m_K!}f^{\sum_k m_k}(b) \prod_{k} \omega_k^{m_k} 
2\delta_{x_k+m_k,e} \nonumber\\
&=&
\sum_{m_k \ge 0}
\frac{1}{(x_1+2m_1)!\cdots (x_K+2m_K)!}
f^{\sum_k x_k+\sum_k 2m_k}(b) \prod_{k} \omega_k^{x_k+2m_k},\label{eq:fphi}
\end{eqnarray}
\end{widetext}
by using $\sum_{z}z^n = 1 + (-1)^n = 2\delta_{n,e}$, where we introduce
a shorthand notation $\delta_{n,e}$ to represent that the value is 1 if 
and only if $n$ is an even integer. Introducing a parameter $\epsilon$ for each $\omega_k$ to count the order, the asymptotic behavior of $\hat{f}$ for small $\epsilon$ goes as
\begin{eqnarray}
\hat{f}^{x_1 x_2\cdots x_K} =
f^{\sum_k x_k}(b) \epsilon^{\sum_k x_k}
\prod_k \omega_k^{x_k} + O(\epsilon^{\sum_k x_k+2}).\label{eq:asymp}
\end{eqnarray}

Using Eq. \eqref{eq:fphi}, we can also find the parity
of $\hat{f}^{n_1n_2\cdots n_K}$ with respect to the sign change of $\vec{\omega}$.
If one of its component $\omega_k$ changes sign, with the new vector denoted by $\vec{\omega}'$, the corresponding
$f(b+\vec{\omega}^{\prime T}\vec{z})$ has an expansion with coefficients
$(-1)^{x_k}\hat{f}^{x_1 x_2\cdots x_K}$.

\subsection{Introduction of intermediate functions
$\tilde{\Psi}(\vec{n})$ and $\tilde{\Psi}_+(\vec{n})$}\label{sec:uat}
While Eq. \eqref{eq:z2n} adopts distinct notations ($\Psi$ and $\Phi$) to emphasize the mathematical difference between functions with different types of arguments ($\vec{n}$ or $\vec{z}$), in subsequent sections we will use the same symbol for the wavefunction for the sake of brevity. Whether the representation $\Phi(\vec{z})$ or $\Psi(\vec{n})$ is intended will be clear from the context.

The wavefunction $\Psi(\vec{n})$ can contain zeros and negative values. We first show that we can construct an intermediate state $\tilde{\Psi}$, which is arbitrarily close to $\Psi$ but with nonzero entries. Suppose the number of
zeros in $\Psi$ is $M_0$, if $M_0=0$, then we simply choose $\tilde{\Psi}\triangleq \Psi$. Otherwise, if $M_0 \ge 1$, we  define
\begin{eqnarray}
\tilde{\Psi}(\vec{n})\triangleq \left\{\begin{array}{cc}
  \epsilon/\sqrt{M_0}, & \Psi(\vec{n}) = 0  \\
  \Psi(\vec{n})\sqrt{1-\epsilon^2},  &  \Psi(\vec{n})\ne 0
\end{array}\right.,   
\end{eqnarray}
which is normalized $\langle\tilde{\Psi}|\tilde{\Psi}\rangle=1$ and satisfies
\begin{eqnarray}
|\Psi(\vec{n})-\tilde{\Psi}(\vec{n})|=\left\{\begin{array}{cc}
  \epsilon/\sqrt{M_0}, & \Psi(\vec{n}) = 0  \\
  |\Psi(\vec{n})|(1-\sqrt{1-\epsilon^2}),  &  \Psi(\vec{n})\ne 0
\end{array}\right..    
\end{eqnarray}
In both case, we have $|\Psi(\vec{n})-\tilde{\Psi}(\vec{n})|<\epsilon$ using
$|\Psi(\vec{n})|\le 1$ and $1-\sqrt{1-\epsilon^2}<\epsilon$ for $\epsilon<1$.

Next, we show the following theorem holds.
\begin{theorem}
$\tilde{\Psi}(\vec{n})$ can be written as
\begin{eqnarray}
\tilde{\Psi}(\vec{n}) = s(\vec{n})\tilde{\Psi}_+(\vec{n}),\quad
s(\vec{n}) = \prod_{\alpha=1}^{M_-}
\phi(b_\alpha + \vec{w}_\alpha^T\vec{n}) 
\end{eqnarray}
where $s(\vec{n})$ has the same sign structure as $\tilde{\Psi}(\vec{n})$, and hence $\tilde{\Psi}_+(\vec{n})$ is a function with all
positive values.
\end{theorem}

\begin{proof}
Let the index set $I_- =\{ i : \tilde{\Psi}(\vec{n}_i)<0\}$
contains all the indices where $\tilde{\Psi}(\vec{n}_i)$ is negative and $M_-$ be the number of negative values in $\tilde{\Psi}(\vec{n})$, that is, $M_-=|I_-|$. 
We show that for each $i\in I_-$,
we can construct a factor $\phi(b_\alpha + \vec{w}_\alpha^T\vec{n})$ such that
$\phi(b_\alpha + \vec{w}_\alpha^T\vec{n}_i)<0$
and 
$\phi(b_\alpha + \vec{w}_\alpha^T\vec{n}_j)>0$
for $j\ne i$.
Without loss of generality, we can assume $\phi(x) > 0$ for $x \in (x_0 - \epsilon, x_0)$ and $\phi(x) < 0$ for $x \in (x_0, x_0 + \epsilon)$ based on the first condition in Theorem \ref{thm:NPSgeneral}. The case with reversed signs follows similarly.

By the hyperplane separation theorem\cite{boyd2004convex} applied to the hypercube vertices $\{\vec{n}_j\}_{j=0}^{D_K-1}$, there exists $\vec{u} \in \mathbb{R}^K$ and $c \in \mathbb{R}$ such that:
\begin{eqnarray}
\vec{u}^T \vec{n}_i + c & > & 0, \\
\vec{u}^T \vec{n}_j + c & < & 0, \quad \text{for all } j \neq i.
\end{eqnarray}
Define the affine function
\begin{eqnarray}
L(\vec{n}) & = & \theta(\vec{u}^T \vec{n} + c) + x_0,
\end{eqnarray}
where $\theta > 0$ is a scaling parameter to be determined. Let
\begin{eqnarray}
M & = & \vec{u}^T \vec{n}_i + c > 0, \\
m & = & \max_{j \neq i} |\vec{u}^T \vec{n}_j + c| > 0.
\end{eqnarray}
Choose $\theta$ such that:
\begin{eqnarray}
\theta < \frac{\epsilon}{\max(M, m)},
\end{eqnarray}
then for all configurations,
\begin{eqnarray}
|L(\vec{n}_j) - x_0| & = & \theta|\vec{u}^T \vec{n}_j + c| < \epsilon \quad \text{for all } j,
\end{eqnarray}
that is,
\begin{eqnarray}
L(\vec{n}_i) & \in & (x_0, x_0 + \epsilon), \\
L(\vec{n}_j) & \in & (x_0 - \epsilon, x_0), \quad \text{for all } j \neq i.
\end{eqnarray}
By the sign properties of $\phi$ around $x_0$, we obtain
\begin{eqnarray}
\phi(L(\vec{n}_i)) & < & 0, \\
\phi(L(\vec{n}_j)) & > & 0 \quad \text{for all } j \neq i.
\end{eqnarray}
Setting $\vec{w} = \theta\vec{u}$ and $b = \theta c + x_0$ completes the proof.
\end{proof}

This proof demonstrates that we can construct individual factors that control the sign pattern of NPS, which is a crucial step in establishing the universal approximation capability of NPS.
In the following, we prove that $\tilde{\Psi}_+(\vec{n})$
can also be approximate arbitrarily well using NPS with only positive factors.

\subsection{Recursive approximation of $\tilde{\Psi}_+(\vec{n})$ by NPS}
To approximate $\tilde{\Psi}_+(\vec{n})$ using NPS, the second condition in Theorem \ref{thm:NPSgeneral} is important.
Its necessity can be seen easily as follows:
If $\kappa(x)=\ln\phi(x)$
is a polynomial with degree less than $K$, which implies that
$\phi(x)=e^{P_n(x)}$ with $P_n(x)=\sum_{i=0}^{n}p_i x^i$ being a polynomial in $x$ with
degree $n<K$, then the product of two neurons becomes
\begin{eqnarray}
\phi(b_1+\vec{w}_1^T\vec{z})\phi(b_2+\vec{w}_2^T\vec{z}) 
&=&
e^{P_n(b_1+\vec{w}_1^T\vec{z})+P_n(b_2+\vec{w}_2^T\vec{z})} \nonumber\\
&=& e^{Q_n(z_1,z_2,\cdots,z_K)},
\end{eqnarray}
where $Q_n(z_1,z_2,\cdots,z_K)$ is a multilinear polynomial in $\{z_k\}$
with degree $n<K$. This reveals that the product of $\phi$
does not have the representational power to approximate
terms with degree higher than $n$, e.g., $e^{g z_1z_2\cdots z_K}$.

The sufficiency can be proved as follows. Specifically, we will show that if $\kappa(x)=\ln\phi(x)$ is not a polynomial in $x$ with degree less than $K$, then 
there exists a set of parameters $\{N_i,b_i,\vec{w}_i\}_{i=1}^{D_K-1}$ 
such that $\tilde{\Psi}(\vec{z})$ can be approximated arbitrarily well by NPS, viz.,
\begin{eqnarray}
|\tilde{\Psi}_+(\vec{z}) - \Theta(\vec{z})|<\epsilon,\label{eq:nps+}
\end{eqnarray}
with
\begin{eqnarray}
\Theta(\vec{z})
=
\mathcal{N}\prod_{\alpha=1}^{D_K-1} \phi^{N_\alpha}(b_\alpha+\vec{w}_\alpha^T \vec{z}),
\label{eq:phiprod}
\end{eqnarray}
for any given $\epsilon$, where $\mathcal{N}>0$ is a normalization constant and $\phi(b_\alpha+\vec{w}_\alpha^T \vec{z})>0$.
In the following, we show that 
\begin{eqnarray}
|\ln\tilde{\Psi}_+(\vec{z})-\ln\Theta(\vec{z})|<\epsilon,\label{eq:matchln}
\end{eqnarray}
by matching the multilinear coefficient of $\ln\tilde{\Psi}_+(\vec{z})$ using neurons in 
$\ln\Theta(\vec{z})$ in a recursive way. Then, by the continuity of $e^x$, Eq. \eqref{eq:nps+} holds.

To this end, we order the power set of $I=\{1,2,\cdots,K\}$ into $K+1$ tiers
\begin{eqnarray}
\left\{\begin{array}{lll}
\mathrm{tier}\;K &:&\; \{1,\cdots,K\} \nonumber\\
\mathrm{tier}\;K-1 &:&\; \{1,\cdots,K-1\}, \cdots, \{2,\cdots,K\} \nonumber\\
 &&\quad\cdots \nonumber\\
\mathrm{tier}\;2 &:&\; \{1,2\},\cdots,\{K-1,K\} \nonumber\\
\mathrm{tier}\;1 &:&\; \{1\}, \{2\}, \cdots, \{K\} \nonumber\\
\mathrm{tier}\;0 &:&\; \varnothing.
\end{array}\right.\label{eq:tier}
\end{eqnarray}
For brevity, we define the symbol $I_{t,k}$ to represent the $k$-th subset at the tier $t$. Let $g(\vec{z})=\ln\tilde{\Psi}_+(\vec{z})$, then using the multilinear expansion
\eqref{eq:multilinear} for $g(\vec{z})$ gives
\begin{eqnarray}
g(\vec{z}) = \hat{g}_{I_{0}}
+
\sum_{k=1}^{C_K^{1}} \hat{g}_{I_{1,k}}\mathcal{Z}_{I_{1,k}}
+
\cdots
+
\hat{g}_{I_{K}}\mathcal{Z}_{I_{K}},\label{eq:gexpansion}
\end{eqnarray}
or equivalently
\begin{eqnarray}
\tilde{\Psi}_+(\vec{z}) = \exp\left(
\hat{g}_{I_{0}}
+
\sum_{k=1}^{C_K^{1}} \hat{g}_{I_{1,k}}\mathcal{Z}_{I_{1,k}}
+
\cdots
+
\hat{g}_{I_{K}}\mathcal{Z}_{I_{K}}
\right),
\end{eqnarray}
where we have introduced an abbreviation $\mathcal{Z}_{I_{t,k}}$
for the factor $z_{p_1}z_{p_2}\cdots z_{p_t}$ with $p_1<p_2<\cdots <p_t$ and $p_i\in I_{t,k}$,
and $\hat{g}_{I_{t,k}}$ to represent the Fourier
coefficient $\hat{g}^{x_1,\cdots,x_K}$ with $x_k=1$
for $k\in I_{t,k}$.

Eq. \eqref{eq:phiprod} gives
\begin{eqnarray}
\ln\Theta(\vec{z}) 
&=& 
\sum_{\alpha=1}^{D_K-1} N_\alpha \ln\phi(b_\alpha+\vec{w}_\alpha^T\vec{z}) + \ln\mathcal{N} \nonumber\\
&=&
\sum_{\alpha=1}^{D_K-1} N_\alpha\kappa(b_\alpha+\vec{w}_\alpha^T\vec{z}) + \ln\mathcal{N}.
\end{eqnarray}
The basic idea for proving Eq. \eqref{eq:matchln}
is to use one term in this expansion to match
the term with the highest power of multilinear polynomials, such as
Eq. \eqref{eq:gexpansion}, at each time in a recursive way starting from the highest tier in Eq. \eqref{eq:tier}. Specifically, we
use the following form
\begin{eqnarray}
\ln\Theta(\vec{z}) = \sum_{I_{t,k},t>0} N_{I_{n,k}}\kappa(b_{I_{n,k}}+\vec{w}_{I_{n,k}}^T\vec{z}_{I_{n,k}})
+ \ln\mathcal{N},\label{eq:phin2}
\end{eqnarray}
where the sum is replaced by the sum over subset from tier $K$ to tier 1,
and $\vec{z}_{I_{t,k}}$ represents the vector consisted of 
$z_{p_i}$ with $p_i\in I_{t,k}$. This structure is illustrated
in a schematic way in Fig. \ref{fig:Phi}. We will
show how to choose $\{N_{I_{t,k}},b_{I_{t,k}},\omega_{I_{t,k}}\}$
such that Eq. \eqref{eq:matchln} holds.

\begin{figure}
    \centering
    \includegraphics[width=0.45\textwidth]{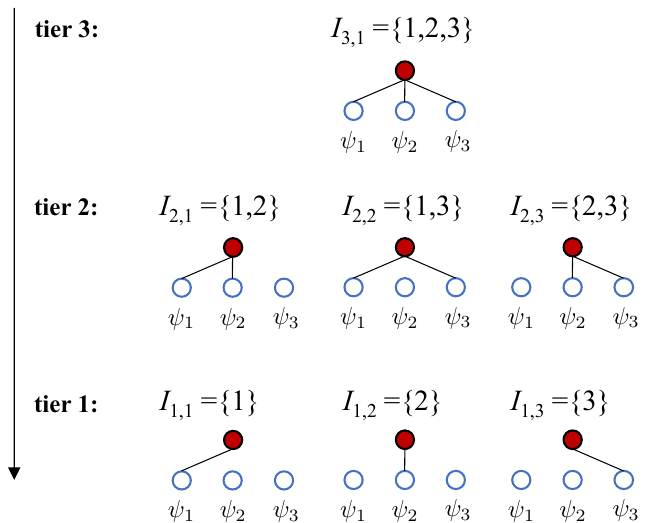}
    \caption{Hierarchical structure of neurons $\kappa(b_{I_{t,k}}+
    \vec{\omega}^T_{I_{t,k}}\vec{z}_{I_{t,k}})$
    used to eliminate the leading terms in $\mathcal{Z}_{I_{t,k}}$ recursively.
    Specifically, we can match the coefficients in Eq. \eqref{eq:gexpansion}
    from the highest tier to the lowest tier recursively using neurons with supports
    shown in the figure.
    }
    \label{fig:Phi}
\end{figure}

We start by giving a recipe to approximate the higher degree term. We rewrite Eq. \eqref{eq:gexpansion} as
\begin{eqnarray}
g(\vec{z}) \triangleq g^{(0)}_{K}(z_{I_K}) = g^{(0)}_{K-1}(z_{I_{K}}) + \hat{g}^{(0)}_{I_{K}}\mathcal{Z}_{I_{K}},\label{eq:ghighest}
\end{eqnarray}
where $g^{(0)}_{K-1}(z_{I_{K}})$ collects the remaining 
$2^{K}-1$ terms is the multilinear polynomial 
in variables $z_{I_K}=(z_1,\cdots,z_K)$, with the subscript 
$K-1$ indicating the highest degree. 
Now we choose one term in
Eq. \eqref{eq:phin2}, and rewritten it as
\begin{eqnarray}
\kappa(b_{I_K}+\vec{\omega}_{I_K}^T\vec{z}_{I_K})
=\kappa_{K-1}(z_{I_K}) + \hat{\kappa}_{I_K} \mathcal{Z}_{I_K},
\end{eqnarray}
where $\kappa_{K-1}(z_{I_K})$ and $\hat{\kappa}_{I_K}$ are both functions of $(b_{I_K},\vec{\omega}_{I_K})$, but for brevity
we have omitted their function dependence. Our goal is to use it to match the term with the highest power $\hat{g}^{(0)}_{I_{K}}\mathcal{Z}_{I_{K}}$
in Eq. \eqref{eq:ghighest} via
\begin{eqnarray}
&&
g^{(0)}_{K}(z_{I_K}) - N_{I_{K}}\kappa(b_{I_{K}}+\vec{w}^T_{I_{K}}\vec{z}_{I_{K}}) \nonumber\\
&=&
[g^{(0)}_{K-1}(z_{I_{K}})-N_{I_{K}}\kappa_{K-1}(z_{I_{K}})]
+ (\hat{g}^{(0)}_{I_{K}} - N_{I_{K}} \hat{\kappa}_{I_{K}})\mathcal{Z}_{I_{K}} \nonumber\\
&\triangleq&
g_{K-1}^{(1)}(z_{I_{K}}) + (\hat{g}^{(0)}_{I_{K}} - N_{I_{K}} \hat{\kappa}_{I_{K}})\mathcal{Z}_{I_{K}},
\end{eqnarray}    
where $g_{K-1}^{(1)}(z_{I_{K}})$ is a new polynomial with highest power
equal to $K-1$. In a similar way, we can use $C_{K}^{K-1}$ terms
in Eq. \eqref{eq:phin2} to match the terms with power equal to $K-1$ in
$g_{K-1}^{(1)}(z_{I_{K}})$, viz.,
{\small\begin{align}
&
g_{K-1}^{(1)}(z_{I_{K}}) - \sum_{k=1}^{C_K^{K-1}} N_{I_{K-1,k}}\kappa(b_{I_{K-1,k}}+\vec{w}^T_{I_{K-1,k}}\vec{z}_{I_{K-1,k}}) \nonumber\\
&=
[g_{K-2}^{(1)}(z_{I_K})
-\sum_{k=1}^{C_K^{K-1}}N_{I_{K-1,k}}\kappa_{K-2}(z_{I_{K-1,k}})] \nonumber\\
&+
\sum_{k=1}^{C_K^{K-1}}(\hat{g}^{(1)}_{I_{K-1,k}} - N_{I_{K-1,k}} \hat{\kappa}_{I_{K-1,k}})\mathcal{Z}_{I_{K-1,k}}, \nonumber\\
&\triangleq
g_{K-2}^{(2)}(z_{I_K})
+
\sum_{k=1}^{C_K^{K-1}}(\hat{g}^{(1)}_{I_{K-1,k}} - N_{I_{K-1,k}} \hat{\kappa}_{I_{K-1,k}})\mathcal{Z}_{I_{K-1,k}}, 
\end{align}}
This procedure can be carried out recursively down to tier 1, such that
{\small\begin{eqnarray}
&&
g_{1}^{(K-1)}(z_{I_K}) - \sum_{k=1}^{C_K^{1}} N_{I_{1,k}}\kappa(b_{I_{1,k}}+\vec{w}^T_{I_{1,k}}\vec{z}_{I_{1,k}}) \nonumber\\
&=&
[g_{0}^{(K-1)}(z_{I_K})
-\sum_{k=1}^{C_K^1}N_{I_{1,k}}\kappa_{0}(z_{I_{1,k}})] \nonumber\\
&&+
\sum_{k=1}^{C_K^1}(\hat{g}^{(K-1)}_{I_{1,k}} - N_{I_{1,k}} \hat{\kappa}_{I_{1,k}})\mathcal{Z}_{I_{1,k}}, \nonumber\\
&=&
g_{0}^{(K)}(z_{I_K})
+
\sum_{k=1}^{C_K^1}(\hat{g}^{(K-1)}_{I_{1,k}} - N_{I_{2,k}} \hat{\kappa}_{I_{2,k}})\mathcal{Z}_{I_{1,k}}.
\end{eqnarray}}Now $g_{0}^{(K)}(z_{I_K})$ is just a constant and we can simply choose $\ln \mathcal{N}=g_{0}^{(K)}(z_{I_K})$
to exactly match it.

Summing up these equations from tier $K$ to tier 1 leads to 
\begin{eqnarray}
&&
\ln\tilde{\Psi}_+(\vec{z}) - \ln\Theta(\vec{z}) \nonumber\\
&=&
\sum_{t=1}^{K}
\sum_{k=1}^{C_K^t} (\hat{g}^{(t)}_{I_{K-t,k}} - N_{I_{K-t,k}} \hat{\kappa}_{I_{K-t,k}})\mathcal{Z}_{I_{K-t,k}}
\end{eqnarray}
with $\ln\Theta(\vec{z})=\sum_{t=1}^{K}
\sum_{k=1}^{C_K^t} N_{I_{K-t,k}}\kappa(b_{I_{K-t,k}}+\vec{w}^T_{I_{K-t,k}}\vec{z}_{I_{K-t,k}})+\ln\mathcal{N}$.
Suppose by choosing $\{N_{I_{K-t,k}},b_{I_{K-t,k}},\vec{\omega}_{I_{K-t,k}}\}$, we can make
\begin{eqnarray}
|\hat{g}^{(n)}_{I_{K-t,k}} - N_{I_{K-t,k}} \hat{\kappa}_{I_{K-t,k}}| \le (\epsilon/K)^{K-t},\label{eq:approx}
\end{eqnarray}
for an arbitrarily small $\epsilon$, then
{\small
\begin{eqnarray}
&&|\ln\tilde{\Psi}_+(\vec{z}) - \ln\Theta(\vec{z})|\nonumber\\
&\le&
\sum_{t=1}^{K}
\sum_{k=1}^{C_K^t} |\hat{g}^{(n)}_{I_{K-t,k}} - N_{I_{K-t,k}} \hat{\kappa}_{I_{K-t,k}}|\cdot|\mathcal{Z}_{I_{K-t,k}}| \nonumber\\
&\le&
\sum_{t=1}^{K}
C_{K}^t (\epsilon/K)^{K-t}
=
(1+\epsilon/K)^K-1
\le 2\epsilon,
\end{eqnarray}}for $0\le \epsilon<1$ and using $|\mathcal{Z}_{I_{K-t,k}}|=1$.
This show that if Eq. \eqref{eq:approx} holds, then
by an recursive procedure, we can construct an approximation to $\ln\tilde{\Psi}_+(\vec{z})$ arbitrarily well using NPS.

\subsection{Proof of Eq. \eqref{eq:approx}}
We show that for any tolerance $ \delta > 0 $, there exist integer $ N $, and real parameters $ b $ and $ \vec{\omega}$ such that
\begin{eqnarray}
\left| \hat{g} - N \hat{\kappa}(b, \vec{\omega}) \right| \leq \delta,\label{eq:approx2}
\end{eqnarray}
which proves Eq. \eqref{eq:approx}. For simplicity, the subscripts for symbols in Eq. \eqref{eq:approx} have been omitted. The key observation is that the Fourier coefficients $\hat{\kappa}(b, \vec{\omega})$ can be arbitrarily small,
see Eq. \eqref{eq:asymp}. Specifically, from the small-$\omega$ expansion of $ \hat{\kappa}(b, \vec{\omega})$, we know that for small $\vec{\omega}$, the Fourier coefficient behaves as:
\begin{eqnarray}
\hat{\kappa}(b, \vec{\omega}) &=& A(b) \prod_k \omega_k^{x_k} + O\left( \|\omega\|^{\sum x_k + 2} \right),
\end{eqnarray}
where $ A(b) $ is an analytic function of $ b $. 
As $\vec{\omega} \to 0 $, $ \hat{\kappa}(b, \vec{\omega})$ goes to 0 continuously, which means that we can make $ \hat{\kappa}(b, \omega) $ arbitrarily small by choosing $\vec{\omega}$ sufficiently small. [If $A(b)\equiv 0$, we can always use the higher order term whose coefficient is nonzero. This is possible as the assumption is that $\kappa(x)$ is not a polynomial of degree less than $K$.]

For simplicity, we assume that $\hat{g}$ and $\hat{\kappa}(b,\vec{\omega})$ are of the same sign, otherwise, we can flip the sign of $\hat{\kappa}(b,\vec{\omega})$ by making
the sign of one $\omega_k$ with $x_k=1$ in $\vec{\omega}$ 
flipped, see the discussion after Eq. \eqref{eq:asymp}.
Now, we define the integer $N =\operatorname{round}\left( \frac{\hat{g}}{\hat{\kappa}(b,\vec{\omega})}\right)$,
then the error due to rounding is given by
\begin{eqnarray}
\left| \frac{\hat{g}}{\hat{\kappa}(b,\vec{\omega})} - N \right|
\leq \frac{1}{2},
\end{eqnarray}
or equivalently,
\begin{eqnarray}
\left| \hat{g} - N \hat{\kappa}(b, \vec{\omega}) \right|
\leq \frac{|\hat{\kappa}(b, \vec{\omega})|}{2}.
\end{eqnarray}
Since $\hat{\kappa}(b, \vec{\omega})$ can be made arbitrarily small by choosing $\vec{\omega}$ sufficiently small, we can ensure that the error is less than $ \delta $ by choosing
$|\hat{\kappa}(b, \vec{\omega})| \le \delta$ using small
enough $\vec{\omega}$, which completes the proof of Eq. \eqref{eq:approx2}.

\bibliography{main}

\end{document}